\documentclass{cmip-l}
\usepackage{amsmath,amsfonts,amssymb,amsthm}
\usepackage{feynmp,psfrag,graphicx}

\input xy
\xyoption{all}

\newtheorem*{thm-ward}{Theorem \ref{thm:ward}}%[section]
\newtheorem{thm}{Theorem}%[section]
\newtheorem{corl}[thm]{Corollary}
\newtheorem{lma}[thm]{Lemma} 
\newtheorem{prop}[thm]{Proposition}
\newtheorem{defn}[thm]{Definition}

\newtheorem{rem}[thm]{Remark}

\def\Aslash{A\mkern-11.5mu/\,}
\def\Aut{\mathrm{Aut}}
\def\bar{\overline}

\def\C{\mathbb{C}}

\def\CK{\textup{CK}}

\def\dirac{\partial\mkern-11.5mu/\,}
\def\Diff{\bar{\textup{Diff}}}

\def\F{\mathcal{F}}
\def\g{\mathfrak{g}}

\def\gh{\textup{gh}}

\def\Hom{\mathrm{Hom}}
\def\half{\tfrac{1}{2}}

\def\Loc{\textup{Loc}}

\def\nn{\nonumber}

\def\P{\mathcal{P}}
\def\R{\mathbb{R}}

\def\res{\mathrm{res}}

\def\su{\mathfrak{su}}
\def\SU{\mathrm{SU}}
\def\Sym{\mathrm{Sym}}

\def\tr{\mathrm{tr~}}

\def\Z{\mathbb{Z}}

\def\gho{
{~
      \begin{fmfgraph}(8,6)
	\fmfpen{.1mm}
	\fmfset{dot_len}{.5mm}
        \fmfforce{(0w,.15h)}{l}
        \fmfforce{(1w,.15h)}{r}
        \fmf{dots}{l,r}
    \end{fmfgraph}
}}

\def\glu{
{~
      \begin{fmfchar}(8,6)
	\fmfpen{.1mm}
	\fmfset{curly_len}{.5mm}
	\fmfleft{l}
	\fmfright{r}
	\fmf{gluon}{l,r}
      \end{fmfchar}
}}

\def\ghoglu{
{~
      \begin{fmfchar}(10,8)
	\fmfpen{.1mm}
	\fmfset{curly_len}{.5mm}
	\fmfset{dot_len}{.5mm}
	\fmfleft{l}
	\fmfright{r1,r2}
	\fmf{gluon}{l,v}
	\fmf{dots}{r1,v}
	\fmf{dots}{v,r2}
      \end{fmfchar}
}}

\title[Renormalization Hopf algebras and BRST-symmetries]{Renormalization Hopf algebras \\ for gauge theories and BRST-symmetries}
\author{Walter D. van Suijlekom}
\address{Institute for Mathematics, Astrophysics and Particle Physics\\
Faculty of Science, Radboud University Nijmegen, 
Heyendaalseweg 135, 6525 AJ Nijmegen, The Netherlands}
\email{waltervs@math.ru.nl}
\date{\today}

\begin{document}
\begin{fmffile}{graphs-boston}

\fmfset{wiggly_len}{5pt} % reducing length of wiggles in boson lines
\fmfset{wiggly_slope}{70} % increasing slope of wiggles in boson lines
\fmfset{curly_len}{3pt} % reducing length of wiggles in boson lines
\fmfset{curly_len}{1.2mm}

\fmfset{dot_len}{1mm}

\begin{abstract}
The structure of the Connes--Kreimer renormalization Hopf algebra is studied for gauge theories, with particular emphasis on the BRST-formalism. We work in the explicit example of quantum chromodynamics, the physical theory of quarks and gluons. 

A coaction of the renormalization Hopf algebra is defined on the coupling constants and the fields. In this context, BRST-invariance of the action implies the existence of certain Hopf ideals in the renormalization Hopf algebra, encoding the Slavnov--Taylor identities for the coupling constants.
\end{abstract}
\maketitle

\section{Introduction}
\label{sect:intro}

Quantum gauge field theories are most successfully described perturbatively, expanding around the free quantum field theory. In fact, at present its non-perturbative formulation seems to be far beyond reach so it is the only thing we have. On the one hand, many rigorous results can be obtained \cite{BBH94,BBH95} using cohomological arguments within the context of the BRST-formalism \cite{BRS74,BRS75,BRS76,Tyu75}. On the other hand, renormalization of perturbative quantum field theories has been carefully structured using Hopf algebras \cite{Kre98,CK99,CK00}. The presence of a gauge symmetry induces a rich additional structure on these Hopf algebras, as has been explored in \cite{Kre05,KY06,BKUY08} and in the author's own work \cite{Sui07,Sui07c,Sui08}. All of this work is based on the algebraic transparency of BPHZ-renormalization, with the Hopf algebra reflecting the recursive nature of this procedure. 

In this article we study more closely the relation between the renormalization Hopf algebras and the BRST-symmetries for gauge theories. We work in the explicit case of quantum chromodynamics (QCD), a Yang--Mills gauge theory with gauge group $\SU(3)$ that describes the strong interaction between quarks and gluons. We will shortly describe this in a little more detail, as well as the appearance of BRST-symmetries. 

After describing the renormalization Hopf algebra for QCD, we study its structure in Section \ref{sect:hopf}. The link between this Hopf algebra and the BRST-symmetries acting on the fields is established in Section \ref{sect:coaction}.

\section*{Acknowledgements}
The author would like to thank the participants of the workshop ``DIAMANT meets GQT'' at the Lorentz Center in Leiden.

\section{Quantum chromodynamics}
\label{sect:ym}

In order to keep the discussion in this article as explicit as possible, we will work in the setting of quantum chromodynamics (QCD). 
This is an example of a Yang--Mills gauge theory, as introduced originally in \cite{YM54}. It is the physical theory that successfully describes the so-called strong interaction between quarks and gluons. Let us make more precise how these particles can be described mathematically, at least in a perturbative approach. 

One of the basic principles in the dictionary between the (elementary particle) physicists' and mathematicians' terminology is that 
\begin{center}
``particles are representations of a Lie group.''
\end{center}
In the case of quantum chromodynamics, this Lie group -- generally called the gauge group -- is $\SU(3)$. In fact, the {\it quark} is a $\C^3$-valued function $\psi = (\psi_i)$ on spacetime $M$. This `fiber' $\C^3$ at each point of spacetime is the defining representation of $\SU(3)$. Thus, there is an action on $\psi$ of an $\SU(3)$-valued function on $M$; let us write this function as $U$, so that $U(x) \in \SU(3)$. 
In physics, the three components of $\psi$ correspond to the so-called {\it color} of the quark, typically indicated by red, green and blue. 

The {\it gluon}, on the other hand, is described by an $\su(3)$-valued one-form on $M$, that is, a section of $\Lambda^1 (\su(3)) \equiv \Lambda^1 \otimes (M \times \su(3))$. We have in components 
$$
A=A_\mu dx^\mu = A_\mu^a dx^\mu T^a
$$ 
where the $\{ T^a \}_{a=1}^8$ form a basis for $\su(3)$. The structure constants $\{ f^{ab}_c \}$ of $\g$ are defined by $[T^a, T^b]=f^{ab}_c T^c$ and the normalization is such that $\tr(T^a T^b) = \delta^{ab}$. It is useful to think of $A$ as a connection one-form (albeit on the trivial bundle $M \times \SU(3)$).
The group $\SU(3)$ acts on the second component $\su(3)$ in the adjoint representation. Again, this is pointwise on $M$, leading to an action of $U=U(x)$ on $A$. In both cases, that is, for quarks and gluons, the transformations
\begin{equation}
\psi_i \mapsto U_{ij} \psi_j, \qquad  A_\mu \mapsto g^{-1} U^{-1} \partial_\mu U +  U^{-1} A_\mu U
\end{equation}
are called {\it gauge transformations}. The constant $g$ is the so-called {\it strong coupling constant}. 

As in mathematics, also in physics one is after {\it invariants}, in this case, one looks for functions -- or, rather, functionals -- of the quark and gluon fields that are invariant under a local ({i.e.} $x$-dependent) action of $\SU(3)$. We are interested in the following action functional:
\begin{equation}
\label{eq:ym}
S(A,\psi) = \frac{1}{8 \pi} \int_M F_{\mu\nu}^a F^{\mu\nu}_a  + \bar\psi_i (i \gamma^\mu \partial_\mu + \gamma^\mu A_{\mu}^a T^a_{ij} + m )  \psi_j 
\end{equation}
with $F \equiv F(A):= d A + g A^2$ the curvature of $A$; it is an $\su(3)$-valued 2-form on $M$. 
Before checking that this is indeed invariant under $\SU(3)$, let us explain the notation in the last term. The $\gamma^\mu$ are the Dirac matrices, and satisfy
$$
\gamma^\mu \gamma^\nu + \gamma^\nu \gamma^\mu = -2 \delta^{\mu\nu}
$$
Clearly, this relation cannot be satisfied by complex numbers (which are never anti-commuting). In fact, the representation theory of the algebra with the above relation ({i.e.} the Clifford algebra) is quite rich. The idea is that the fields $\psi$ are not only $\C^3$-valued, but that actually each of the components $\psi_i$ is itself a 4-vector, called {\it spinors}. This is so as to accommodate a representation of the Clifford algebra: in 4 spacetime dimensions the Dirac matrices are 4-dimensional (although in general this dimension is $2^{[n/2]}$ for $n$ spacetime dimensions). Besides this matrix multiplication, the partial derivative $\partial_\mu$ acts componentswise, as does the {\it mass} $m$ which is really just a real number. Finally, for our purposes it is sufficient to think of $\bar\psi$ as the (componentswise) complex conjugate of $\psi$. The typical Grassmannian nature of these fermionic fields is only present in the current setup through the corresponding {\it Grassmann degree} of $+1$ and $-1$ that is assigned to both of them.

Introducing the notation $\dirac = \gamma^\mu\partial_\mu$ and $\Aslash = \gamma^\mu A_\mu$, we can write
$$
S(A,\psi) = - \big\langle F(A) , F(A) \big\rangle +  \big\langle \psi,( i \dirac + \Aslash +m )\psi \big\rangle,
$$
in terms of appropriate inner products. Essentially, these are combinations of spinorial and Lie algebra traces and the inner product on differential forms.
For more details, refer to the lectures by Faddeev in \cite{DelEA99}. The key observation is that the $\SU(3)$-valued functions $U(x)$ act by unitaries with respect to this inner product. 

\subsection{Ghost fields and BRST-quantization}

In a path integral quantization of the field theory defined by the above action, one faces the following problem. Since gauge transformations are supposed to act as symmetries on the theory, the gauge degrees of freedom are irrelevant to the final physical outcome. Thus, in one way or another, one has to quotient by the group of gauge transformations. However, gauge transformations are $\SU(3)$-valued function on $M$, yielding an infinite dimensional group. In order to deal with this infinite redundancy, Faddeev and Popov used the following trick. They introduced so-called {\it ghost fields}, denoted by $\omega$ and $\bar\omega$. In the case of quantum chromodynamics, these are $\su(3)[-1]$ and $\su(3)[1]$-valued functions on $M$, respectively. The shift $[-1]$ and $[+1]$ is to denote that $\omega$ and $\bar\omega$ have {\it ghost degree} $1$ and $-1$, respectively. Consequently, they have {\it Grassmann degree} $1$ and $-1$, respectively. 
In components, we write
$$
\omega = \omega^a T^a; \qquad
\bar\omega = \bar\omega^a T^a. 
$$ 
Finally, an auxiliary field $h$ -- also known as the Nakanishi--Lautrup field -- is introduced; it is an $\su(3)$-valued function (in ghost degree 0) and we write $h = h^a T^a$. 

The dynamics of the ghost fields and their interaction with the gauge field are described by the rather complicated additional term:
$$
S_{\gh}(A,\omega,\bar\omega,h) = - \big\langle A, dh \big\rangle + \big \langle d \bar\omega, d \omega \big\rangle
+ \frac{1}{2} \xi \big \langle h ,h  \big\rangle
+ g \big \langle d \bar\omega, [A,\omega] \big\rangle,
$$
where $\xi \in \R$ is the so-called {\it gauge parameter}.

The essential point about the ghost fields is that, in a path integral formulation of quantum gauge field theories, their introduction miraculously takes care of the {\it fixing of the gauge}, {i.e.} picking a point in the orbit in the space of fields under the action of the group of gauge transformations. The ghost fields are the ingredients in the BRST-formulation that was developed later by Becchi, Rouet, Stora and independently by Tyutin in \cite{BRS74,BRS75,BRS76,Tyu75}. Let us briefly describe this formalism.

Because the gauge has been fixed by adding the term $S_{\gh}$, the combination $S + S_{\gh}$ is not invariant any longer under the gauge transformations. This is of course precisely the point. Nevertheless, $S+ S_{\gh}$ possesses another symmetry, which is the BRST-symmetry. It acts on function(al)s in the fields as a ghost degree $1$ derivation $s$, which is defined on the generators by
\begin{gather}
\label{brst}
s A = d \omega +g [A,\omega],\qquad
s \omega = \frac{1}{2} g [\omega,\omega],\qquad
s \bar\omega =  h\\ \nn
s h =0\qquad
s \psi =  g \omega \psi , \qquad s \bar\psi = g \bar\psi \omega.
\end{gather}
Indeed, one can check (eg., see \cite[Sect. 15.7]{Wei96} for details) that $s(S + S_{\gh})=0$.

The form degree and Grassmann degree of the fields are combined in the {\it total degree} and summarized in the following table:
\begin{center}
\begin{tabular}{|l|r|r|r|r|r|r|}
\hline
& $A$ & $\omega$ & $\bar\omega$ & $h$& $\psi$ &$\bar\psi$ \\ \hline
Grassmann degree 	&0    &$+1$ &$-1$ &0 &$+1$ &$-1$ 	\\ \hline
%fermionic degree 	&0    &$0$ &$0$ &0 &$+1$ &$+1$	\\ \hline
form degree 	&$+1$ &$0$  &$0$  &0 &0 &0 	\\ \hline
total degree	&$+1$ &$+1$ &$-1$ &0 &$+1$ &$-1$	\\ \hline
\end{tabular}
\end{center}

The fields generate an algebra, the algebra of local forms $\Loc(\Phi)$. With respect to the above degrees, it decomposes as before into $\Loc^{(p,q)}(\Phi)$ with $p$ the form degree and $q$ the Grassmann degree. The total degree is then $p+q$ and $\Loc(\Phi)$ is a graded Lie algebra by setting 
$$
[ X, Y ] = X Y - (-1)^{\deg(X)\deg(Y)} Y X,
$$
with the grading given by this total degree. Note that the present graded Lie bracket is of degree $0$ with respect to the total degree, that is, $\deg([X,Y]) = \deg(X) + \deg(Y)$. It satisfies graded skew-symmetry, the graded Leibniz identity and the graded Jacobi identity:
\begin{align*}
&[X,Y] = - (-1)^{\deg(X)\deg(Y)} [Y,X], \\ 
& [XY,Z] = X [ Y,Z] + (-1)^{\deg(Y)\deg(Z)} [X,Z]Y.\\
&(-1)^{\deg(X)\deg(Z)} [ [ X,Y],Z] + (\hbox{cyclic perm.}) %(-1)^{\deg(X)\deg(Y)} [[Y,Z],X] + (-1)^{\deg(Y)\deg(Z)} [[Z,X],Y]
= 0,
\end{align*}

\begin{lma}
The BRST-differential, together with the above bracket, gives $\Loc(\Phi)$ the structure of a graded differential Lie algebra.
\end{lma}
Moreover, the BRST-differential $s$ and the exterior derivative $d$ form a double complex, that is, $d \circ s + s \circ d=0$ and
$$
\xymatrix{ 
&\vdots  & \vdots & \vdots &\\
&\Loc^{(0,1)} \ar[u]_s \ar[r]_d  &\Loc^{(1,1)} \ar[u]_s \ar[r]_d  &\Loc^{(2,1)}\ar[u]_s \ar[r]_d & \cdots\\
&\Loc^{(0,0)} \ar[u]_s \ar[r]_d  &\Loc^{(1,0)} \ar[u]_s \ar[r]_d  &\Loc^{(2,0)}\ar[u]_s \ar[r]_d & \cdots\\
&\Loc^{(0,-1)} \ar[u]_s \ar[r]_d  &\Loc^{(1,-1)} \ar[u]_s \ar[r]_d  &\Loc^{(2,-1)}\ar[u]_s \ar[r]_d & \cdots\\
& \vdots \ar[u]_s   & \vdots \ar[u]_s  & \vdots \ar[u]_s & 
}
$$
This double complex has a quite interesting structure in itself, and was the subject of study in \cite{BBH94,BBH95}. This contained further applications in renormalization and the description of anomalies.

\section{Renormalization Hopf algebra for QCD}
\label{sect:hopf}
As we discussed previously, quantum chromodynamics describes the interaction between quarks and gluons. In order to do this successfully at a quantum level, it was necessary to introduce ghost fields. We will now describe how the dynamics and interaction of and between these fields, naturally give rise to Feynman graphs. These constitute a Hopf algebra which encodes the procedure of renormalization in QCD. We will describe this Hopf algebra, and study its structure in terms of the so-called Green's functions. 

\subsection{Hopf algebra of Feynman graphs}

First of all, the quark, ghost and gluon fields are supposed to {\it propagate}, this we will denote by a straight, dotted and curly line or {\it edges} as follows:
\begin{align*}
e_1 =~
\parbox{25pt}{
  \begin{fmfgraph}(20,10)
      \fmfleft{l}
      \fmflabel{}{l}
      \fmfright{r}
      \fmf{plain}{l,r}
  \end{fmfgraph}
}
\qquad 
e_2 = 
\parbox{25pt}{
  \begin{fmfgraph}(20,10)
      \fmfleft{l}
      \fmflabel{}{l}
      \fmfright{r}
      \fmf{dots}{l,r}
  \end{fmfgraph}
}
\qquad
e_3 =~
\parbox{25pt}{
  \begin{fmfgraph}(20,10)
      \fmfleft{l}
      \fmflabel{}{l}
      \fmfright{r}
      \fmf{gluon}{l,r}
  \end{fmfgraph}
}.
\end{align*}
The interactions between the fields then naturally appear as {\it vertices}, connecting the edges corresponding to the interacting fields. The allowed interactions in QCD are the following four:
\begin{align*}
v_1 =~\parbox{35pt}{
    \begin{fmfchar}(25,25)
      \fmfleft{l}
      \fmfright{r1,r2}
      \fmf{gluon}{l,v}
      \fmf{plain}{r1,v}
      \fmf{plain}{v,r2}
      \fmfdot{v}
    \end{fmfchar}},
\qquad
v_2 =~\parbox{35pt}{
  \begin{fmfgraph}(25,25)
      \fmfleft{l}
      \fmfright{r1,r2}
      \fmf{gluon}{l,v}
      \fmf{dots}{r1,v}
      \fmf{dots}{v,r2}
      \fmfdot{v}
  \end{fmfgraph}},
\qquad
v_3 =~\parbox{35pt}{
  \begin{fmfgraph}(25,25)
    \fmfleft{l}
      \fmfright{r1,r2}
      \fmf{gluon}{l,v}
      \fmf{gluon}{r1,v}
      \fmf{gluon}{v,r2}
      \fmfdot{v}
  \end{fmfgraph}},
\qquad
v_4 =~
\parbox{35pt}{
  \begin{fmfchar}(25,25)
    \fmfleft{l1,l2}
      \fmfright{r1,r2}
      \fmf{gluon}{l1,v}
      \fmf{gluon}{l2,v}
      \fmf{gluon}{r1,v}
      \fmf{gluon}{v,r2}
      \fmfdot{v}
  \end{fmfchar}}.
\end{align*}
In addition, since the quark is supposed to have a mass, there is a {\it mass term}, which we depict as a vertex of valence two:
$$
v_5 =~\parbox{25pt}{ \begin{fmfgraph}(25,10)
      \fmfleft{l}
      \fmfright{r}
      \fmf{plain}{l,v}
      \fmf{plain}{v,r}
      \fmfdot{v}
    \end{fmfgraph}}~.
$$

We can make the relation between these edges (vertices) and the propagation (interaction) more precise through the definition of a map $\iota$ that assigns to each of the above edges and vertices a (monomial) functional in the fields. In fact, the assignment $e_i \mapsto \iota(e_i)$ and $v_j \mapsto \iota(v_j)$ is
\begin{center}
\begin{figure}[h!]
\label{fig:setR}
\begin{tabular}{|l|ccc|ccccc|}
\hline
& $e_1$ & $e_2$ & $e_3$ & $v_1$ & $v_2$ & $v_3$ & $v_4$ & $v_5$\\[2mm]
edge/vertex &\parbox{15pt}{
  \begin{fmfchar}(15,10)
      \fmfleft{l}
      \fmflabel{}{l}
      \fmfright{r}
      \fmf{plain}{l,r}
  \end{fmfchar}}
&
\parbox{20pt}{
  \begin{fmfchar}(15,10)
      \fmfleft{l}
      \fmflabel{}{l}
      \fmfright{r}
      \fmf{dots}{l,r}
  \end{fmfchar}
}
&
\parbox{20pt}{
  \begin{fmfchar}(15,10)
      \fmfleft{l}
      \fmflabel{}{l}
      \fmfright{r}
      \fmf{gluon}{l,r}
  \end{fmfchar}
}
&
\parbox{15pt}{
    \begin{fmfchar}(15,10)
      \fmfleft{l}
      \fmfright{r1,r2}
      \fmf{gluon}{l,v}
      \fmf{plain}{r1,v}
      \fmf{plain}{v,r2}
    \end{fmfchar}}
& \parbox{15pt}{
  \begin{fmfgraph}(15,10)
      \fmfleft{l}
      \fmfright{r1,r2}
      \fmf{gluon}{l,v}
      \fmf{dots}{r1,v}
      \fmf{dots}{v,r2}
  \end{fmfgraph}}
&\parbox{15pt}{
  \begin{fmfgraph}(15,10)
    \fmfleft{l}
      \fmfright{r1,r2}
      \fmf{gluon}{l,v}
      \fmf{gluon}{r1,v}
      \fmf{gluon}{v,r2}
  \end{fmfgraph}}
&
\parbox{15pt}{
  \begin{fmfchar}(15,10)
    \fmfleft{l1,l2}
      \fmfright{r1,r2}
      \fmf{gluon}{l1,v}
      \fmf{gluon}{l2,v}
      \fmf{gluon}{r1,v}
      \fmf{gluon}{v,r2}
  \end{fmfchar}}
&
\parbox{15pt}{
\begin{fmfgraph}(15,15)
      \fmfleft{l}
      \fmfright{r}
      \fmf{plain}{l,v}
      \fmf{plain}{v,r}
	\fmfdot{v}
    \end{fmfgraph}}\\[3mm]
\hline
\hline 
monomial $\iota$ &$i \bar\psi \dirac \psi$&$d \bar\omega d \omega$&$dA dA $&$\psi \Aslash \psi$ & $\bar\omega [A,\omega]$  &$2 dA A^2$ & $A^4$ & $m \bar\psi \psi$\\
\hline
\end{tabular}
\caption{QCD edges and vertices, and (schematically) the corresponding monomials in the fields.}
\end{figure}
\end{center}

\begin{rem}
We have not assigned an edge to the field $h$; this is because it does not interact with any of the other fields. Its only -- still crucial -- effect is on the propagator of the gluon, through the terms $-\langle A, dh \rangle$ and $\half \xi\langle h,h\rangle$.
\end{rem}

A {\it Feynman graph} is a graph built from these vertices and edges. Naturally, we demand edges to be connected to vertices in a compatible way, respecting their straight, dotted or curly character. As opposed to the usual definition in graph theory, Feynman graphs have no external vertices. However, they do have {\it external lines} which come from vertices in $\Gamma$ for which some of the attached lines remain vacant ({i.e.} no edge attached).

If a Feynman graph $\Gamma$ has two external quark (straight) lines, we would like to distinguish between the propagator and mass terms. Mathematically, this is due to the presence of the vertex of valence two. 
In more mathematical terms, since we have vertices of valence two, we would like to indicate whether a graph with two external lines corresponds to such a vertex, or to an edge. A graph $\Gamma$ with two external lines is dressed by a bullet when it corresponds to a vertex, {i.e.} we write $\Gamma_\bullet$. The above correspondence between Feynman graphs and vertices/edges is given by the {\it residue} $\res(\Gamma)$. 
It is defined for a general graph as the vertex or edge it corresponds to after collapsing all its internal points. 
For example, we have:
\begin{gather*}
\res\left( \parbox{40pt}{
\begin{fmfgraph*}(40,30)
      \fmfleft{l}
      \fmfright{r1,r2}
      \fmf{gluon}{l,v}
      \fmf{plain}{v,v1,r1}
      \fmf{plain}{v,v2,r2}
      \fmffreeze
      \fmf{gluon}{v1,loop,v2}
      \fmffreeze
      \fmfv{decor.shape=circle, decor.filled=0, decor.size=2thick}{loop}
\end{fmfgraph*}
}\right) = 
\parbox{20pt}{
\begin{fmfgraph*}(20,20)
      \fmfleft{l}
      \fmfright{r1,r2}
      \fmf{gluon}{l,v}
      \fmf{plain}{v,r1}
      \fmf{plain}{v,r2}
\end{fmfgraph*}
}
\qquad \text{ and }\qquad
\res\left( \parbox{40pt}{
\begin{fmfgraph*}(40,30)
      \fmfleft{l}
      \fmfright{r}
      \fmf{plain}{l,v1,v2,v5,v6,r}
      \fmf{gluon,right,tension=0}{v5,v1}
      \fmf{gluon,right,tension=0}{v2,v6}
\end{fmfgraph*}
}\right) = 
\parbox{20pt}{
\begin{fmfgraph*}(20,20)
      \fmfleft{l}
      \fmfright{r}
      \fmf{plain}{l,r}
\end{fmfgraph*}}
\intertext{ but }
\res\left( \parbox{40pt}{
\begin{fmfgraph*}(40,30)
      \fmfleft{l}
      \fmfright{r}
      \fmf{plain}{l,v1,v2,v5,v6,r}
      \fmf{gluon,right,tension=0}{v5,v1}
      \fmf{gluon,right,tension=0}{v2,v6}
\end{fmfgraph*}}{}_\bullet
\right) = 
\parbox{20pt}{
\begin{fmfgraph*}(20,20)
      \fmfleft{l}
      \fmfright{r}
      \fmf{plain}{l,v,r}
      \fmfdot{v}
\end{fmfgraph*}}~.
\end{gather*}

For the definition of the Hopf algebra of Feynman graphs, we restrict to {\it one-particle irreducible} (1PI) Feynman graphs. These are graphs that are not trees and cannot be disconnected by cutting a single internal edge. 
\begin{defn}[Connes--Kreimer \cite{CK99}]
The Hopf algebra $H_{\CK}$ of Feynman graphs is the free commutative algebra over $\C$ generated by all 1PI Feynman graphs with residue in $R= R_V \cup R_E$, with counit $\epsilon(\Gamma)=0$ unless $\Gamma=\emptyset$, in which case $\epsilon(\emptyset)=1$, and coproduct
\begin{align*}
\Delta (\Gamma) = \Gamma \otimes 1 + 1 \otimes \Gamma + \sum_{\gamma \subsetneq \Gamma} \gamma \otimes \Gamma/\gamma,
\end{align*}
where the sum is over disjoint unions of 1PI subgraphs with residue in $R$. The quotient $\Gamma/\gamma$ is defined to be the graph $\Gamma$ with the connected components of the subgraph contracted to the corresponding vertex/edge. If a connected component $\gamma'$ of $\gamma$ has two external lines, then there are possibly two contributions corresponding to the valence two vertex and the edge; the sum involves the two terms $\gamma'_\bullet \otimes \Gamma/(\gamma' \to \bullet)$ and $\gamma' \otimes \Gamma/\gamma'$.
The antipode is given recursively by
\begin{equation}
\label{antipode}
S(\Gamma) = - \Gamma - \sum_{\gamma \subsetneq \Gamma} S(\gamma) \Gamma/\gamma.
\end{equation}
\end{defn}
Two examples of this coproduct are:\\[2mm]
\begin{align*}
\Delta(
\parbox{35pt}{\begin{fmfgraph*}(35,11)
      \fmfleft{l}
      \fmfright{r}
      \fmf{plain}{l,v1,v2,v3,v4,r}
      \fmf{gluon,right,tension=0}{v4,v1}
      \fmf{gluon,right,tension=0}{v3,v2}
\end{fmfgraph*}}
) &= 
\parbox{35pt}{\begin{fmfgraph*}(35,11)
      \fmfleft{l}
      \fmfright{r}
      \fmf{plain}{l,v1,v2,v3,v4,r}
      \fmf{gluon,right,tension=0}{v4,v1}
      \fmf{gluon,right,tension=0}{v3,v2}
\end{fmfgraph*}}
 \otimes 1 + 1 \otimes 
\parbox{35pt}{\begin{fmfgraph*}(35,11)
      \fmfleft{l}
      \fmfright{r}
      \fmf{plain}{l,v1,v2,v3,v4,r}
      \fmf{gluon,right,tension=0}{v4,v1}
      \fmf{gluon,right,tension=0}{v3,v2}
\end{fmfgraph*}}
+
\parbox{30pt}{
\begin{fmfgraph*}(30,11)
      \fmfleft{l}
      \fmfright{r}
      \fmf{plain}{l,v2,v3,r}
      \fmf{gluon,right,tension=0}{v3,v2}
\end{fmfgraph*}}
\otimes
\parbox{30pt}{
\begin{fmfgraph*}(30,11)
      \fmfleft{l}
      \fmfright{r}
      \fmf{plain}{l,v1,v4,r}
      \fmf{gluon,right,tension=0}{v4,v1}
\end{fmfgraph*}}
+
\parbox{30pt}{
\begin{fmfgraph*}(30,11)
      \fmfleft{l}
      \fmfright{r}
      \fmf{plain}{l,v2,v3,r}
      \fmf{gluon,right,tension=0}{v3,v2}
\end{fmfgraph*}}
{ }_\bullet
\otimes
\parbox{30pt}{
\begin{fmfgraph*}(30,11)
      \fmfleft{l}
      \fmfright{r}
      \fmf{plain}{l,v1,v3,v4,r}
      \fmfdot{v3}
      \fmf{gluon,right,tension=0}{v4,v1}
\end{fmfgraph*}}~,
\\[3mm]
\Delta(
  \parbox{35pt}{
    \begin{fmfgraph*}(35,35)
      \fmfleft{l}
      \fmfright{r}
      \fmf{phantom}{l,v1,v2,r}
      \fmf{gluon}{l,v1}
      \fmf{gluon}{v2,r}
      \fmf{phantom,left,tension=0,tag=1}{v1,v2}
      \fmf{phantom,right,tension=0,tag=2}{v1,v2}
      \fmffreeze
      \fmfi{plain}{subpath (0,.8) of vpath1(__v1,__v2)}
      \fmfi{plain}{subpath (0,.8) of vpath2(__v1,__v2)}
      \fmfi{plain}{subpath (0.8,1.2) of vpath1(__v1,__v2)}
      \fmfi{plain}{subpath (0.8,1.2) of vpath2(__v1,__v2)}
      \fmfi{gluon}{point .8 of vpath1(__v1,__v2) .. point .8 of vpath2(__v1,__v2)}
      \fmfi{gluon}{point 1.2 of vpath1(__v1,__v2) .. point 1.2 of vpath2(__v1,__v2)}
      \fmfi{plain}{subpath (1.2,2) of vpath1(__v1,__v2)}
      \fmfi{plain}{subpath (1.2,2) of vpath2(__v1,__v2)}
    \end{fmfgraph*}}
) &= 
 \parbox{35pt}{
    \begin{fmfgraph*}(35,35)
      \fmfleft{l}
      \fmfright{r}
      \fmf{phantom}{l,v1,v2,r}
      \fmf{gluon}{l,v1}
      \fmf{gluon}{v2,r}
      \fmf{phantom,left,tension=0,tag=1}{v1,v2}
      \fmf{phantom,right,tension=0,tag=2}{v1,v2}
      \fmffreeze
      \fmfi{plain}{subpath (0,.8) of vpath1(__v1,__v2)}
      \fmfi{plain}{subpath (0,.8) of vpath2(__v1,__v2)}
      \fmfi{plain}{subpath (0.8,1.2) of vpath1(__v1,__v2)}
      \fmfi{plain}{subpath (0.8,1.2) of vpath2(__v1,__v2)}
      \fmfi{gluon}{point .8 of vpath1(__v1,__v2) .. point .8 of vpath2(__v1,__v2)}
      \fmfi{gluon}{point 1.2 of vpath1(__v1,__v2) .. point 1.2 of vpath2(__v1,__v2)}
      \fmfi{plain}{subpath (1.2,2) of vpath1(__v1,__v2)}
      \fmfi{plain}{subpath (1.2,2) of vpath2(__v1,__v2)}
    \end{fmfgraph*}}
 \otimes 1 + 1 \otimes 
\parbox{35pt}{
    \begin{fmfgraph*}(35,35)
      \fmfleft{l}
      \fmfright{r}
      \fmf{phantom}{l,v1,v2,r}
      \fmf{gluon}{l,v1}
      \fmf{gluon}{v2,r}
      \fmf{phantom,left,tension=0,tag=1}{v1,v2}
      \fmf{phantom,right,tension=0,tag=2}{v1,v2}
      \fmffreeze
      \fmfi{plain}{subpath (0,.8) of vpath1(__v1,__v2)}
      \fmfi{plain}{subpath (0,.8) of vpath2(__v1,__v2)}
      \fmfi{plain}{subpath (0.8,1.2) of vpath1(__v1,__v2)}
      \fmfi{plain}{subpath (0.8,1.2) of vpath2(__v1,__v2)}
      \fmfi{gluon}{point .8 of vpath1(__v1,__v2) .. point .8 of vpath2(__v1,__v2)}
      \fmfi{gluon}{point 1.2 of vpath1(__v1,__v2) .. point 1.2 of vpath2(__v1,__v2)}
      \fmfi{plain}{subpath (1.2,2) of vpath1(__v1,__v2)}
      \fmfi{plain}{subpath (1.2,2) of vpath2(__v1,__v2)}
    \end{fmfgraph*}}
+2~ \parbox{35pt}{
  \begin{fmfgraph*}(35,35)
      \fmfleft{l}
      \fmfright{r1,r2}
      \fmf{gluon}{l,v}
      \fmf{plain}{v,v1,v3,r1}
      \fmf{plain}{v,v2,v4,r2}
      \fmffreeze
      \fmf{gluon}{v1,v2}      
      \fmf{gluon}{v3,v4}
  \end{fmfgraph*}}
\otimes 
\parbox{35pt}{
\begin{fmfgraph*}(35,35)
  \fmfleft{l}
  \fmfright{r}
  \fmf{phantom}{l,v1,v2,r}
  \fmf{gluon}{l,v1}
  \fmf{gluon}{v2,r}
  \fmf{plain,left,tension=0}{v1,v2}
  \fmf{plain,left,tension=0}{v2,v1}
\end{fmfgraph*}}\\
& \qquad
+2~ \parbox{35pt}{
  \begin{fmfgraph*}(35,35)
      \fmfleft{l}
      \fmfright{r1,r2}
      \fmf{gluon}{l,v}
      \fmf{plain}{v,v1,r1}
      \fmf{plain}{v,v2,r2}
      \fmffreeze
      \fmf{gluon}{v1,v2}
  \end{fmfgraph*}}
\otimes 
\parbox{35pt}{
    \begin{fmfgraph*}(35,35)
      \fmfleft{l}
      \fmfright{r}
      \fmf{phantom}{l,v1,v2,r}
      \fmf{gluon}{l,v1}
      \fmf{gluon}{v2,r}
      \fmf{phantom,left,tension=0,tag=1}{v1,v2}
      \fmf{phantom,right,tension=0,tag=2}{v1,v2}
      \fmffreeze
      \fmfi{plain}{subpath (0,1) of vpath1(__v1,__v2)}
      \fmfi{plain}{subpath (0,1) of vpath2(__v1,__v2)}
      \fmfi{gluon}{point 1 of vpath1(__v1,__v2) .. point 1 of vpath2(__v1,__v2)}
      \fmfi{plain}{subpath (1,2) of vpath1(__v1,__v2)}
      \fmfi{plain}{subpath (1,2) of vpath2(__v1,__v2)}
    \end{fmfgraph*}
}
+  \parbox{35pt}{
  \begin{fmfgraph*}(35,35)
      \fmfleft{l}
      \fmfright{r1,r2}
      \fmf{gluon}{l,v}
      \fmf{plain}{v,v1,r1}
      \fmf{plain}{v,v2,r2}
      \fmffreeze
      \fmf{gluon}{v1,v2}
  \end{fmfgraph*}}
 \parbox{35pt}{
  \begin{fmfgraph*}(35,35)
      \fmfleft{l}
      \fmfright{r1,r2}
      \fmf{gluon}{l,v}
      \fmf{plain}{v,v1,r1}
      \fmf{plain}{v,v2,r2}
      \fmffreeze
      \fmf{gluon}{v1,v2}
  \end{fmfgraph*}}
\otimes  
\parbox{35pt}{
\begin{fmfgraph*}(35,35)
  \fmfleft{l}
  \fmfright{r}
  \fmf{phantom}{l,v1,v2,r}
  \fmf{gluon}{l,v1}
  \fmf{gluon}{v2,r}
  \fmf{plain,left,tension=0}{v1,v2}
  \fmf{plain,left,tension=0}{v2,v1}
\end{fmfgraph*}}~.
\end{align*}

\bigskip

The above Hopf algebra is an example of a connected graded Hopf algebra: it is graded by the {\it loop number $L(\Gamma)$} of a graph $\Gamma$. Indeed, one checks that the coproduct (and obviously also the product) satisfy the grading by loop number and $H^0_{\CK}$ consists of complex multiples of the empty graph, which is the unit in $H_{\CK}$, so that $H^0_{\CK}=\C 1$. We denote by $q_l$ the projection of $H_{\CK}$ onto $H^l_{\CK}$. 

In addition, there is another grading on this Hopf algebra. It is given by the number of vertices and already appeared in \cite{CK99}. However, since we consider vertices and edges of different types (straight, dotted and curly), we extend to a multigrading as follows. For each vertex $v_j$ ($j=1,\ldots,5$) we define a degree $d_j$ as 
$$
d_j ( \Gamma) = \# \text{vertices } v_j \text{ in } \Gamma - \delta_{v_j, \res(\Gamma)}
$$
The multidegree indexed by $j=1,\ldots,5$ is compatible with the Hopf algebra structure, since contracting a subgraph $\Gamma \mapsto \Gamma/\gamma$ creates a new vertex. With this one easily arrives at the following relation:
\begin{align*}
d_j (\Gamma/\gamma) = d_j (\Gamma) - d_j (\gamma)
\end{align*}
Moreover, $d_j (\Gamma \Gamma') = d_j (\Gamma)+ d_j (\Gamma')$ giving a decomposition as vector spaces:
$$
H_{\CK}= \bigoplus_{(n_1,\ldots,n_5) \in \Z^5} H^{n_1,\ldots,n_5}_{\CK}, 
$$
We denote by $p_{n_1,\ldots,n_5}$ the projection onto $H^{n_1,\ldots,n_5}_{\CK}$. Note that also $H^{0,\ldots, 0}_{\CK}=\C 1$.

\begin{lma}
\label{lma:rel-degrees}
There is the following relation between the grading by loop number and the multigrading by number of vertices:
$$
\sum_{j=1}^5 (N(v_j)-2)d_j = 2 L
$$
where $N(v_j)$ is the valence of the vertex $v_j$.
\end{lma}

\begin{proof}
This can be easily proved by induction on the number of internal edges using invariance of the quantity $\sum_{j} (N(v_j)-2)d_j - 2 L$ under the adjoint of an edge.
\end{proof}

The group $\Hom_\C(H_{\CK},\C)$ dual to $H_{\CK}$ is called the {\it group of diffeographisms} (for QCD). This name was coined in general in \cite{CK00} motivated by its relation with the group of (formal) diffeomorphisms of $\C$ (see Section \ref{sect:coaction} below). Stated more precisely, they constructed a map from the group of diffeographisms to the group of formal diffeomorphisms. We have established this result in general ({i.e.} for any quantum field theory) in \cite{Sui08}. Below, we will make a similar statement for Yang--Mills gauge theories.

\subsection{Birkhoff decomposition}
We now briefly recall how renormalization is an instance of 
%obtain Equation \eqref{bphz} for the renormalized amplitude and the counterterm for a graph as 
a Birkhoff decomposition in the group of characters of $H$ as established in \cite{CK99}. Let us first recall the definition of a Birkhoff decomposition.

We let $\gamma: C \to G$ be a loop with values in an arbitrary complex Lie group $G$, defined on a smooth simple curve $C \subset \P_1(\C)$. Let $C_\pm$ be the two complements of $C$ in $\P_1(\C)$, with $\infty \in C_-$. A {\it Birkhoff decomposition} of $\gamma$  is a factorization of the form 
$$
\gamma(z) = \gamma_-(z)^{-1} \gamma_+(z); \qquad (z \in C),
$$
where $\gamma_\pm$ are (boundary values of) two holomorphic maps on $C_\pm$, respectively, with values in $G$. This decomposition gives a natural way to extract finite values from a divergent expression. Indeed, although $\gamma(z)$ might not holomorphically extend to $C_+$, $\gamma_+(z)$ is clearly finite as $z\to 0$.

Now consider a Feynman graph $\Gamma$ in the Hopf algebra $H_{\CK}$. Via the so-called Feynman rules -- which are dictated by the Lagrangian of the theory -- one associates to $\Gamma$ the Feynman amplitude $U(\Gamma)(z)$. It depends on some regularization parameter, which in the present case is a complex number $z$ (dimensional regularization). The famous divergences of quantum field theory are now `under control' and appear as poles in the Laurent series expansion of $U(\Gamma)(z)$. 

On a curve around $0 \in \P^1(\C)$ we can define a loop $\gamma$ by $\gamma(z)(\Gamma):=U(\Gamma)(z)$ which takes values in the group of diffeographisms $G=\Hom_\C(H_{\CK},\C)$. Connes and Kreimer proved the following general result in \cite{CK99}.
\begin{thm}
Let $H$ be a graded connected commutative Hopf algebra with character group $G$. Then any loop $\gamma: C \to G$ admits a Birkhoff decomposition.
\end{thm}
In fact, an explicit decomposition can be given in terms of the group $G(K)= \Hom_\C(H,K)$ of $K$-valued characters of $H$, where $K$ is the field of convergent Laurent series in $z$.\footnote{In the language of algebraic geometry, there is an affine group scheme $G$ represented by $H$ in the category of commutative algebras. In other words, $G=\Hom_\C(H,~ \cdot ~)$ and $G(K)$ are the $K$-points of the group scheme. } 
If one applies this to the above loop associated to the Feynman rules, the decomposition gives exactly renormalization of the Feynman amplitude $U(\Gamma)$: the map $\gamma_+$ gives the renormalized Feynman amplitude and the $\gamma_-$ provides the counterterm. 

\bigskip

Although the above construction gives a very nice geometrical description of the process of renormalization, it is a bit unphysical in that it relies on individual graphs that generate the Hopf algebra. Rather, in physics the probability amplitudes are computed from the full expansion of Green's functions. Individual graphs do not correspond to physical processes and therefore a natural question to pose is how the Hopf algebra structure behaves at the level of the Green's functions. We will see in the next section that they generate Hopf subalgebras, {i.e.} the coproduct closes on Green's functions. Here the so-called Slavnov--Taylor identities for the couplings will play a prominent role.

\subsection{Structure of the Hopf algebra}
In this subsection, we study the structure of the above Hopf algebra of QCD Feynman graphs. In fact, from a dual point of view, the group of diffeographisms turns out to be related to the group of formal diffeomorphisms of $\C^5$. Moreover, we will establish the existence of Hopf ideals, which correspond on the group level to subgroups.

We define the {\it 1PI Green's functions} by
\begin{equation}
\label{green}
G^{e_i} = 1 - \sum_{\res(\Gamma)=e_i} \frac{\Gamma}{\Sym(\Gamma)},\qquad G^{v_j} = 1 + \sum_{\res(\Gamma)=v_j} \frac{\Gamma}{\Sym(\Gamma)} 
\end{equation}
with $i=1,2,3$ and $j=1,\ldots,5$.
The restriction of the sum to graphs $\Gamma$ at loop order $L(\Gamma)=l$ is denoted by $G^r_l$, with $r \in \{ e_i, v_j\}_{i,j}$.

\begin{rem}
Let us explain the meaning of the inverse of Green's functions in our Hopf algebra. Since any Green's function $G^r$ starts with the identity, we can surely write its inverse formally as a geometric series. Recall that the Hopf algebra is graded by loop number. Hence, the inverse of a Green's function at a fixed loop order is in fact well-defined; it is given by restricting the above formal series expansion to this loop order. More generally, we understand any real power of a Green's function in this manner.
\end{rem}

We state without proof the following result of \cite{Sui08}.
\begin{prop}
\label{prop:cop-green}
The coproduct takes the following form on (real powers of) the Green's functions:
\begin{align*}
%\label{cop-Ge}
\Delta \big( (G^{e_i} )^\alpha \big) &= \sum_{n_1, \ldots, n_5 } (G^{e_i})^\alpha  Y_{v_1}^{n_1} \cdots Y_{v_5}^{n_5}
 \otimes p_n ((G^{e_i} )^\alpha),\\ \nonumber
\Delta ( (G^{v_j} )^\alpha ) &= \sum_{n_1, \ldots, n_5 } (G^{v_j})^\alpha Y_{v_1}^{n_1 } \cdots Y_{v_5}^{n_5}
\otimes p_n ((G^{v_j})^\alpha),
\end{align*}
with $\alpha \in \R$.
Consequently, the algebra $H$ generated by the Green's functions (in each vertex multidegree) $G^{e_i}$ ($i=1,2,3$) and $G^{v_j}$ $(j=1,\ldots,5$) is a Hopf subalgebra of $H_{\CK}$. 
\end{prop}

Denote by $N_k(r)$ the number of edges $e_k$ attached to $r \in \{ e_i,v_j \}_{i,j}$; clearly, the total number of lines attached to $r$ can be written as $N(r)=\sum_{i=1,2,3} N_i(r)$. With this notation, define for each vertex $v_j$ an element in $H$ by the formal expansion:
$$
\label{Yv}
Y_{v_j} := \frac{G^{v_j}}{\prod_{i=1,2,3} \left(G^{e_i}\right)^{N_i(v_j)/2} }.
$$
We remark that alternative generators for the Hopf algebra $H$ are $G^{e_j}$ and $Y_{v_j}$, a fact that we will need later on.

\begin{corl}
\label{corl:cop-Yv}
The coproduct on the elements $Y_v$ is given by
$$
\Delta(Y_{v_j}) = \sum_{n_1, \ldots, n_5} Y_{v_j} Y_{v_1}^{n_1} \cdots Y_{v_5}^{n_5}  \otimes p_{n_1 \cdots n_5} (Y_{v_j}),
$$
where $p_{n_1 \cdots n_5}$ is the projection onto graphs containing $n_k$ vertices $v_k$ ($k=1,\ldots,5$).
\end{corl}
\begin{proof}
This follows directly by an application of the formulas in Proposition \ref{prop:cop-green} to $\Delta(Y_{v_j}) = \Delta(G^{v_j}) \prod_{i=1,2,3} \Delta((G^{e_i})^{-N_i(v_j)/2})$.
\end{proof}
Quite remarkably, this formula coincides with the coproduct in the Hopf algebra dual to the group $\Diff(\C^5,0)$ of formal diffeomorphisms tangent to the identity in $5$ variables, closely related to the Fa\`a di Bruno Hopf algebra (cf. for instance the short review \cite{FGV05}). In other words, the Hopf subalgebra generated by $p_{n_1,\ldots,n_5} (Y_{v_j})$ is dual to (a subgroup of) the group $\Diff(\C^5,0)$. This will be further explored in Section \ref{sect:coaction} below.

\begin{corl}{\cite{Sui07c}}
\label{thm:hopfideal}
The ideal $J$ in $H$ generated by $q_l\left(Y_{v_k}^{N(v_j)-2} - Y_{v_j}^{N(v_k)-2}  \right)$ for any $l\geq0$ and $j,k =1,\ldots,4$ is a Hopf ideal, {i.e.}
$$
\Delta(J) \subset J \otimes H + H \otimes J. 
$$
\end{corl}
\begin{proof}
Fix $j$ and $k$ two integers between 1 and 5. Applying the formulas in Proposition \ref{prop:cop-green} to the coproduct on $Y_{v_k}^{N(v_j)-2}$ yields
$$
\Delta\left( Y_{v_k}^{N(v_j)-2} \right) = \sum_{n_1, \ldots, n_5} Y_{v_k}^{N(v_j)-2} Y_{v_1}^{n_1} \cdots Y_{v_5}^{n_5}  \otimes p_{n_1 \cdots n_5} (Y_{v_k}^{N(v_j)-2}),
$$
Now, module elements in $J$, we can write
$$
Y_{v_2}^{n_2} = Y_{v_1}^{n_2 \frac{N(v_2)-2}{N(v_1)-2}}, 
$$
and similarly for $v_3$ and $v_4$ so that
$$
Y_{v_1}^{n_1} \cdots Y_{v_5}^{n_5} = \left(Y_{v_1}^{1/N(v_1)-2} \right)^{\sum_k n_k (N(v_k)-2)} = \left(Y_{v_1}^{1/(N(v_1)-2)}\right)^{2l}.
$$
by an application of Lemma \ref{lma:rel-degrees}. Note that this is independent of the $n_i$ but only depends on the total loop number $l$. For the coproduct, this yields
$$
\Delta\left( Y_{v_k}^{N(v_j)-2} \right) = \sum_{l=0}^\infty Y_{v_k}^{N(v_j)-2} ~Y_{v_1}^{\frac{2l}{N(v_1)-2}} \otimes q_l (Y_{v_k}^{N(v_j)-2}),
$$
Of course, a similar formula holds for the other term defining $J$, upon interchanging $j$ and $k$. For their difference we then obtain
\begin{multline*}
\Delta\left( Y_{v_k}^{N(v_j)-2} - Y_{v_j}^{N(v_k)-2}  \right) =  \sum_{l=0}^\infty \left( Y_{v_k}^{N(v_j)-2} -  Y_{v_j}^{N(v_k)-2} \right) ~Y_{v_1}^{\frac{2l}{N(v_1)-2}}
 \otimes q_l (Y_{v_k}^{N(v_j)-2}) \\ 
+  \sum_{l=0}^\infty Y_{v_k}^{N(v_j)-2}  ~Y_{v_1}^{\frac{2l}{N(v_1)-2}}   \otimes q_l \left( Y_{v_k}^{N(v_j)-2} - Y_{v_j}^{N(v_k)-2} \right).
\end{multline*}
This is an element in $J\otimes H + H\otimes J$, which completes the proof.
\end{proof}
\begin{rem}
\label{rem:generatorsJ}
An equivalent set of generators for $J$ is given by $Y_{v_i} - Y_{v_1}^{N(v_i)-2}$ with $i=2,3,4$. 
\end{rem}

In this Hopf ideal, the reader might have already recognized the Slavnov--Taylor identities for the couplings. Indeed, in the quotient Hopf algebra $H/J$ these identities hold. Moreover, since the character $U : H \to\C$ given by the regularized Feynman rules vanishes on $J$ (these are exactly the Slavnov--Taylor identities) and thus factorizes over this quotient (provided we work with dimensional regularization, or another gauge symmetry preserving regularization scheme). Now, the Birkhoff decomposition for the group $\Hom_\C(H/J,\C)$ gives the counterterm map $C$ and the renormalized map $R$ as characters on $H/J$. Thus, they also satisfy the Slavnov--Taylor identities and this provides a purely algebraic proof of the compatibility of the Slavnov--Taylor identities for the couplings with renormalization, an essential step in proving renormalizablity of gauge theories.

Below, we shall give a more conceptual (rather then combinatorial) explanation for the existence of these Hopf ideals, after establishing a connection between $H$ and the fields and coupling constants.

\section{Coaction and BRST-symmetries}
\label{sect:coaction}
The fact that we encountered diffeomorphism groups starting with Feynman graphs is not very surprising from a physical point of view. Indeed, Feynman graphs are closely involved in the running of the coupling constants described by the renormalization group. In the next subsection, we will clarify this point by defining a coaction of the Hopf algebra $H$ on the coupling constants and the fields. Dually, this will lead to an action of the diffeomorphism group. It contains a subgroup that respects the BRST-invariance of the action, which will be related to the Hopf ideal of the previous section. Finally, its relation with the renormalization group is further described.

\subsection{Coaction on the coupling constants and fields}
In this section, we will establish a connection between the Hopf algebra of Feynman graphs defined above and the fields, coupling constants and masses that characterize the field theory. This allows for a derivation of the Hopf ideals encountered in the previous section from the so-called master equation satisfied by the Lagrangian. 

Let us first introduce formal variables $\lambda_1, \lambda_2, \ldots,\lambda_5$, corresponding to the vertices describing the five possible interactions in QCD. Also, we write $\phi_1=A,~ \phi_2=\psi, ~ \phi_3 = \omega$ and $\phi_4=h$ for the fields, in accordance with the labelling of the edges (see Figure \ref{fig:setR} above). We denote by $\F = \Loc(\phi_1, \phi_2, \phi_3, \phi_4)\otimes \C[[\lambda_1,\ldots,\lambda_5]]$ the algebra of local functionals in the fields $\phi_i$ (and their conjugates), extended linearly by formal power series in the $\lambda_j$. Recall that a local functional is an integral of a polynomial in the fields and their derivatives, and the algebra structure is given by multiplication of these integrals.

\begin{thm}
\label{thm:coaction}
The algebra $\F$ is a comodule algebra over the Hopf algebra $H$
The coaction $\rho: \F \to \F \otimes H$ is given on the generators by
\begin{align*}
\rho : \lambda_j &\longmapsto \sum_{n_1, \ldots, n_5} \lambda_j \lambda_{1}^{n_1 } \cdots \lambda_{5}^{n_5 }  \otimes p_{n_1 \cdots n_5} (Y_{v_j}),\qquad (j=1,\ldots, 5);\\
\rho: \phi_i & \longmapsto \sum_{n_1, \ldots, n_5} \phi~ \lambda_{1}^{n_1 } \cdots \lambda_{5}^{n_5}  \otimes p_{n_1 \cdots n_5} (\sqrt{G^{e_i}}), \qquad (i=1,2,3),
\end{align*}
while it commutes with partial derivatives on $\phi$. 
\end{thm}
\begin{proof}
Since we work with graded Hopf algebras, it suffices to establish that $ (\rho \otimes 1)\circ \rho= (1 \otimes \Delta) \circ \rho$. We claim that this follows from coassociativity ({i.e.} $(\Delta \otimes 1) \circ \Delta = (1\otimes \Delta)\circ \Delta$) of the coproduct $\Delta$ of $H$. Indeed, the first expression very much resembles the form of the coproduct on $Y_j$ as derived in Corollary \ref{corl:cop-Yv}: replacing therein each $Y_{v_k}$ ($k=1, \ldots, 5$) on the first leg of the tensor product by $\lambda_{k}$ and one $\Delta$ by $\rho$ gives the desired result. A similar argument applies to the second expression, using Proposition \ref{prop:cop-green} above. 
\end{proof}

\begin{corl}
The Green's functions $G^{v_j} \in H$ can be obtained when coacting on the interaction monomial $\int \lambda_j \iota(v)(x) d\mu(x)= \int \lambda_j \partial_{\vec\mu_1} \phi_{i_1}(x) \cdots \partial_{\vec\mu_N} \phi_{i_N}(x)d \mu(x)$ for some index set $\{ i_1, \ldots, i_N \}$. 
\end{corl}
For example, 
\begin{align*}
\rho  \bigg( \lambda_2 \langle d \bar\omega, [A,\omega] \rangle \bigg)  &= \sum_{n_1 \cdots n_5} \lambda_2 \lambda_{1}^{n_1} \cdots \lambda_{5}^{n_5}\langle d \bar\omega, [A,\omega] \rangle \otimes  p_{n_1 \cdots n_5} \left(Y_\ghoglu \sqrt{G^\glu} G^\gho \right)
\\
&=\sum_{n_1, \ldots, n_5} \lambda_2 \lambda_{1}^{n_1} \cdots \lambda_{5}^{n_5}\langle d \bar\omega, [A,\omega] \rangle \otimes  p_{n_1 \cdots n_5} \left(G^\ghoglu\right)
\end{align*}

Actually, the first equation in Theorem \ref{thm:coaction} can be interpreted as an action of a subgroup of formal diffeomorphisms in 5 variables on $\C[[\lambda_{1}, \ldots, \lambda_{5}]]$. Let us make this more precise. 

Consider the group $\Diff(\C^5, 0)$ of formal diffeomorphisms in 5 dimensions (coordinates $x_1, \ldots, x_5$) that leave the five axis-hyperplanes invariant. In other words, we consider maps
$$
f(x) = \big( f_1(x ),\ldots, f_5(x) \big)
$$
where each $f_i$ is a formal power series of the form $f_i(x ) = x_i(\sum a^{(i)}_{n_1\cdots n_5}(f) x_1^{n_1} \cdots x_5^{n_5})$ with $a^{(i)}_{0,\ldots,0}=1$ and $x=(x_1, \ldots ,x_5)$. The group multiplication is given by composition, and is conveniently written in a dual manner, in terms of the coordinates. In fact, the $a^{(i)}_{n_1 \cdots n_5}$ generate a Hopf algebra with the coproduct expressed as follows. On the formal generating element $A_i(x) = x_i(\sum a^{(i)}_{n_1\cdots n_k}  x_1^{n_1} \cdots x_k^{n_k})$:
$$
\Delta(A_i(x)) = \sum_{n_1, \ldots, n_k}  A_i(x) \left( A_1(x) \right)^{n_1} \cdots \left( A_k(x) \right)^{n_k} \otimes a^{(i)}_{n_1\cdots n_k}.
$$
Thus, by mapping the $a^{(j)}_{n_1, \ldots, n_5}$ to $p_{n_1,\ldots,n_5}(Y_{v_j})$ in $H$ we obtain a surjective map from $H$ to the Hopf algebra dual to $\Diff(\C^5,0)$. In other words, $\Hom(H,\C)$ is a subgroup of $\Diff(\C^5,0)$ and, in fact, substituting $a^{(j)}_{n_1, \ldots, n_5}$ for $p_{n_1,\ldots,n_5}(Y_{v_j})$ in the first equation of Theorem \ref{thm:coaction} yields (dually) a group action of $\Diff(\C^5,0)$ on $ \C[[\lambda_{1}, \ldots, \lambda_{5}]]$ by $f(a) := (1 \otimes f) \rho(a)$ for $f\in \Diff(\C^5,0)$ and $a \in \C[[\lambda_{1}, \ldots, \lambda_{5}]]$. In fact, we have the following

\begin{prop}
\label{prop:action}
Let $G'$ be the group consisting of algebra maps $f: \F \to \F$ given on the generators by
\begin{align*}
f( \lambda_j)&= \sum_{n_1 \cdots n_5} f^{v_j}_{n_1, \ldots, n_5} \lambda_j \lambda_{1}^{n_1 } \cdots \lambda_{5}^{n_5 }; \qquad (j=1,\ldots, 5) ,\\
f ( \phi_i)&= \sum_{n_1 \cdots n_5} f^{e_i}_{n_1, \ldots, n_5} \phi_i \lambda_{1}^{n_1 } \cdots \lambda_{5}^{n_5 }; \qquad (i =1, \ldots, 3) , \\
\end{align*}
where $f^{v_j}_{n_1 \cdots n_5},f^{e_i}_{n_1 \cdots n_5}\in \C$ are such that $f^{v_j}_{0 \cdots 0} = f^{e_i}_{0\cdots 0} = 1$.
Then the following hold:
\begin{enumerate}
\item The character group $G$ of the Hopf algebra $H$ generated by $p_{n_1\cdots n_5} (Y_v)$ and $p_{n_1\cdots n_5} (\sqrt{G^e})$ with coproduct given in Proposition \ref{prop:cop-green}, is a subgroup of $G'$.
\item The subgroup $N:= \{ f: f(\lambda_j) =\lambda_j, j=1,\ldots, 5 \}$ of $G'$ is normal and isomorphic to $(\C[[\lambda_{1},\ldots,\lambda_{5}]]^\times)^{3}$.
\item $G' \simeq (\C[[\lambda_{1},\ldots,\lambda_{5}]]^\times)^{3} \rtimes \Diff(\C^5,0)$.
\end{enumerate}
\end{prop}
\begin{proof}
From Theorem \ref{thm:coaction}, it follows that a character $\chi$ acts on $\F$ as in the above formula upon writing $f^{v_j}_{n_1 \cdots n_5} = \chi( p_{n_1\cdots n_5} (Y_v) )$ for $j=1,\ldots,5$ and $f^{e_i}_{n_1 \cdots n_5} = \chi( p_{n_1\cdots n_5} (\sqrt{G^{\phi_i}}))$ for $i=1,2,3$.

For (2) one checks by explicit computation that $N$ is indeed normal and that each series $f^{e_i}$ defines an element in $\C[[\lambda_{1}, \ldots, \lambda_{5}]]^\times$ of invertible formal power series. 

Then (3) follows from the existence of a homomorphism from $G'$ to $\Diff(\C^5,0)$. It is given by restricting an element $f$ to $\C[[\lambda_{1}, \ldots, \lambda_{5}]]$. This is clearly the identity map on $\Diff(\C^5,0)$ when considered as a subgroup of $G$ and its kernel is precisely $N$.
\end{proof}

The action of (the subgroup of) $(\C[[\lambda_{1},\ldots,\lambda_{k}]]^\times)^{3} \rtimes \Diff(\C^5,0)$ on $\F$ has a natural physical interpretation: the invertible formal power series act on a field as wave function renormalization whereas the diffeomorphisms act on the coupling constants $\lambda_1,\ldots,\lambda_5$.

\subsection{BRST-symmetries}
We will now show how the previous coaction of the Hopf algebra $H$ on the algebra $\F$ gives rise to the Hopf ideal $J$ encountered before. For this, we choose a distinguished element in $\F$, namely the action $S$. It is given by 
\begin{multline}
\label{eq:action}
S[\phi_i,\lambda_j] = -\big\langle dA , dA \big\rangle 
- 2 \lambda_3 \big\langle dA , A^2 \big\rangle 
- \lambda_4\big\langle A^2 , A^2 \big\rangle 
+  \big\langle \psi,( \dirac + \lambda_1 \Aslash + \lambda_5 )\psi \big\rangle
\\
- \big\langle A, dh \big\rangle + \big \langle d \bar\omega, d \omega \big\rangle
+ \frac{1}{2} \xi \big \langle h ,h  \big\rangle
+ \lambda_2 \big \langle d \bar\omega, [A,\omega] \big\rangle.
\end{multline}
in terms of the appropriate inner products. Note that the action has finitely many terms, that is, it is a (local) polynomial functional in the fields and coupling constants rather than a formal power series.

With the BRST-differential given in Equation \eqref{brst} (involving the `fundamental' coupling constant $g$), we will now impose the BRST-invariance of $S$, by setting
$$
s( S) =0. 
$$
Actually, we will define an ideal $I = \big\langle s (S) \big\rangle$ in $\F$ that implements the relations between the $\lambda_j$'s. Strictly speaking, the fundamental coupling $g$ is not an element in $\F$; we will instead set $g \equiv \lambda_1$. The remaining `coupling' constant $\lambda_5$ is interpreted as the quark mass $m$. 
\begin{prop}
The ideal $I$ is generated by the following elements:
$$
\lambda_1 - \lambda_2; \qquad \lambda_2 - \lambda_3; \qquad \lambda_3 - \lambda_4^2.
$$
\end{prop}
\begin{proof}
This follows directly by applying $s$ (involving $g$) to the action $S$. 
\end{proof}
A convenient set of (equivalent) generators for the ideal $I$ is $\lambda_i - g^{N(v_i)-2}$ for $i=1,\ldots 4$. 
Thus, the image of $S$ in the quotient $\F/I$ is BRST-invariant, that is, $s(S)$ is identically zero. 

Let us return to the group $G \subset (\C[[\lambda_{1},\ldots,\lambda_{5}]]^\times)^{3} \rtimes \Diff(\C^5,0)$, acting on $\F$. Consider the subgroup $G^I$ of $G$ consisting of elements $f$ that leave invariant the ideal $I$, i.e., such that $f(I) \subseteq I$. It is clear from the above generators of $I$ that this will involve a diffeomorphism group in 2 variables, instead of 5. More precisely, we have the following
\begin{thm}[\cite{Sui08}]
\label{thm:groupGI}
Let $J$ be the ideal from Theorem \ref{thm:hopfideal}.
\begin{enumerate}
\item The group $G^I$ acts on the quotient algebra $\F/I$.
\item The image of $G^I$ in $\Aut( \F/I)$ is isomorphic to $\Hom_\C(H/J, \C)$ and $H/J$ coacts on $\F/I$.
\end{enumerate}
Consequently, (the image in $\Aut(\F/I)$ of) $G^I$ is a subgroup of the semidirect product $(\C[[g,\lambda_{5}]]^\times)^{3} \rtimes \Diff(\C^2,0)$. 
\end{thm}
\begin{proof}
The first claim is direct. For the second, note that an element $f \in G$ acts on the generators of $I$ as
\begin{multline*}
f\left( \lambda_{i} -g^{N(v_i)-2} \right) \\ =  \sum_{n_1, \ldots, n_5} \lambda_{1}^{n_1} \cdots\lambda_{5}^{n_5} \left[  \lambda_i f\left( p_{n_1\cdots n_5}(Y_{v_i}) \right) - g^{N(v_i)-2}  f\left( p_{n_1\cdots n_5}(Y_{v_1}^{N(v_i)-2})\right)\right],
\end{multline*}
since $g\equiv \lambda_1$. 
We will reduce this expression by replacing $\lambda_{i}$ by $g^{N(v_i)-2}$, modulo terms in $I$. Together with Lemma \ref{lma:rel-degrees} this yields
$$
f\left( \lambda_{i} -g^{N(v_i)-2} \right) = \sum_{l=0}^\infty g^{2l + N(v_i)-2} ~ f \left(q_l\left( Y_{v_i} - Y_{v_1}^{N(v_i)-2}\right) \right)  \mod I.
$$
The requirement that this is an element in $I$ is equivalent to the requirement that $f$ vanishes on $q_l ( Y_{v_i} - Y_{v_1}^{N(v_i)-2})$, {i.e.} on the generators of $J$, establishing the desired isomorphism. One then easily computes
$$
\rho(I) \subset I \otimes H + \F \otimes J
$$
so that $H/J$ coacts on $\F$ by projecting onto the two quotient algebras.
\end{proof}

In fact, the last claim of the above Theorem can be strengthened. Focusing on the subgroup of the formal diffeomorphism group $\Diff (\C^5, 0)^I$ that leaves invariant the ideal $I$ we have:
$$
1 \to (1+I)^5 \to \Diff (\C^5, 0)^I \to \Diff(\C^2,0) \to 1.
$$
Here, an element $(1+B_i)_{i=1,\ldots,5}$ in $(1+I)^5$ acts on the generators $\lambda_1, \ldots, \lambda_5$ by right multiplication. 
This sequence actually splits, leading to a full description of the group $\Diff (\C^5, 0)^I$. Indeed, by the simple structure of the ideal $I$, a one-sided inverse of the map $\Diff (\C^5, 0)^I \to \Diff(\C^2,0)$ can be easily constructed.
A similar statement holds for the above subgroup $G^I$ of the semidirect product $G\simeq (\C[[\lambda_{1},\ldots,\lambda_{5}]]^\times)^{3} \rtimes \Diff(\C^5,0)$.

In any case, the contents of Theorem \ref{thm:groupGI} have a very nice physical interpretation: the invertible formal power series act on the three fields as wave function renormalization whereas the diffeomorphisms act on one fundamental coupling constant $g$. We will appreciate this even more in the next section where we discuss the renormalization group flow.

\subsection{Renormalization group}
We will now establish a connection between the group of diffeographisms and the renormalization group \`a la Gell-Mann and Low \cite{GL54}. This group describes the dependence of the renormalized amplitudes $\phi_+(z)$ on a mass scale that is implicit in the renormalization procedure. In fact, in dimensional regularization, in order to keep the loop integrals $d^{4-z} k$ dimensionless for complex $z$, one introduces a factor of $\mu^z$ in front of them, where $\mu$ has dimension of mass and is called the {\it unit of mass}. For a Feynman graph $\Gamma$, Lemma \ref{lma:rel-degrees} shows that this factor equals $\mu^{z \sum_{i} (N(v_i)-2)) \delta_{v_i}(\Gamma)/2}$ reflecting the fact that the coupling constants appearing in the action get replaced by 
$$
\lambda_{i} \mapsto \mu^{z \sum_{k} (N(v_k)-2))/2}\lambda_{i}
$$
for every vertex $v_i$ ($i=1,\ldots, 5$). 

As before, the Feynman rules define a loop $\gamma_\mu: C \to G\equiv G(\C)$, which now depends on the mass scale $\mu$. Consequently, there is a Birkhoff decomposition for each $\mu$:
$$
\gamma_\mu(z) = \gamma_{\mu,-}(z)^{-1} \gamma_{\mu,+}(z); \qquad (z \in C),
$$
As was shown in \cite{CK00}, the negative part $\gamma_{\mu,-}(z)$ of this Birkhoff decomposition is independent of the mass scale, that is,
$$
\frac{\partial}{\partial \mu} \gamma_{\mu,-}(z) = 0. 
$$
Hence, we can drop the index $\mu$ and write $\gamma_{-}(z):=\gamma_{\mu,-}(z)$. In terms of the generator $\theta_t$ for the one-parameter subgroup of $G(K)$ corresponding to the grading $l$ on $H$, we can write 
$$
\gamma_{e^t \mu (z) } = \theta_{tz} \left(\gamma_\mu(z) \right), \qquad (t \in \R).
$$
A proof of this and the following result can be found in \cite{CK00}.

\begin{prop}
The limit 
$$
F_t := \lim_{z \to 0} \gamma_-(z) \theta_{tz} \left( \gamma_-(z)^{-1} \right)
$$
exists and defines a $1$-parameter subgroup of $G$ which depends polynomially on $t$ when evaluated on an element $X \in H$. 
\end{prop}
In physics, this 1-parameter subgroup goes under the name of {\it renormalization group}. In fact, using the Birkhoff decomposition, we can as well write
$$
\gamma_{e^t \mu, +}(0) = F_t ~ \gamma_{\mu,+}(0), \qquad (t \in \R). 
$$
This can be formulated in terms of the generator $\beta := \frac{d}{dt} F_t |_{t=0}$ of this 1-parameter group as
\begin{equation}
\label{eq:beta}
\mu \frac{\partial}{\partial \mu} \gamma_{\mu,+}(0) = \beta \gamma_{\mu,+}(0). 
\end{equation}
Let us now establish that this is indeed the beta-function familiar from physics by exploring how it acts on the coupling constants $\lambda_{i}$. First of all, although the name might suggest otherwise, the coupling constants depend on the energy or mass scale $\mu$. Recall the action of $G$ on $\C[[\lambda_{1}, \ldots, \lambda_{5}]]$ defined in the previous section. In the case of $\gamma_{\mu,+}(0) \in G$, we define the (renormalized) {\it coupling constant at scale $\mu$} to be
$$
\lambda_{i}(\mu) = \gamma_{\mu,+}(0)(\lambda_{i}). 
$$
This function of $\mu$ (with coefficients in $\C[[\lambda_1,\ldots, \lambda_5]]$) satisfies the following differential equation:
\begin{equation*}
\beta \left( \lambda_{i}(\mu) \right) = \mu \frac{\partial}{\partial \mu} \left(\lambda_{i}(\mu) \right)
\end{equation*}
which follows easily from Eq. \eqref{eq:beta}. This is exactly the renormalization group equation expressing the flow of the coupling constants $\lambda_{i}$ as a function of the energy scale $\mu$. 
Moreover, if we extend $\beta$ by linearity to the action $S$ of Eq. \eqref{eq:action}, we obtain Wilson's continuous renormalization equation \cite{Wil75}:
$$
\beta(S(\mu)) = \mu \frac{\partial}{\partial \mu} \left( S(\mu) \right)
$$
This equation has been explored in the context of renormalization Hopf algebras in \cite{GKM04, KM08}.

Equation \eqref{eq:beta} expresses $\beta$ completely in terms of $\gamma_{\mu,+}$; as we will now demonstrate, this allows us to derive that for QCD all $\beta$-functions coincide. First, recall that the maps $\gamma_{\mu}$ are the Feynman rules dictated by $S$ in the presence of the mass scale $\mu$, which we suppose to be BRST-invariant: $s(S)=0$. In other words, we are in the quotient of $\F$ by $I = \langle s(S) \rangle$. If the regularization procedure respects gauge invariance, it is well-known that the Feynman amplitude satisfy the Slavnov--Taylor identities for the couplings. In terms of the ideal $J$ defined in the previous section, this means that $\gamma_{\mu} (J)=0$. Since $J$ is a Hopf ideal (Theorem \ref{thm:hopfideal}), it follows that both $\gamma_{\mu,-}$ and $\gamma_{\mu,+}$ vanish on $J$. Indeed, the character $\gamma$ given by the Feynman rules factorizes through $H/J$, for which the Birkhoff decomposition gives two characters $\gamma_+$ and $\gamma_-$ of $H/J$. In other words, if the unrenormalized Feynman amplitudes given by $\gamma_\mu$ satisfy the Slavnov--Taylor identities, so do the counterterms and the renormalized Feynman amplitudes. 

In particular, from Eq. \eqref{eq:beta} we conclude that $\beta$ vanishes on the ideal $I$ in $\C[[\lambda_{1}, \ldots, \lambda_{5}]]$. This implies the following result, which is well-known in the physics literature:
\begin{prop}
All (QCD) $\beta$-functions (for $i=1,\ldots,4$) are expressed in terms of $\beta(g)$ for the fundamental coupling constant $g$:
$$\beta(\lambda_{i}) = \beta(g^{N(v)-2}).
$$
\end{prop}

%\bibliography{references}
\newcommand{\noopsort}[1]{}

\end{fmffile}
\end{document}